\documentclass[10pt,a4paper]{article}
\usepackage{jheppub}
\usepackage{amsmath}
\usepackage[most]{tcolorbox}
\usepackage{dsfont}
\usepackage{ulem}
\usepackage{natbib}
\usepackage{xcolor}
\usepackage[hang,flushmargin]{footmisc}
\usepackage{tikz-cd}
\usepackage{enumitem}
\usepackage{amsfonts,amssymb, amscd,amsmath,latexsym,amsbsy,bm}
\usepackage{stmaryrd}
\usepackage{todonotes}
\usepackage{stackrel}
\usepackage{amsthm}
\newtheorem{theorem}{Theorem}[section]
\newtheorem{lemma}[theorem]{Lemma}
\newtheorem{definition}{Definition}[section]
\usepackage{graphicx}

\title{Bailey pairs for the q-hypergeometric integral pentagon identity}

\author{Ilmar Gahramanov$^{a,b,c}$, and Osman Erkan Kaluc$^{a}$}
\affiliation{
	$^a$ {Department of Physics, Bogazici University, 34342 Bebek, Istanbul, Turkey}\\[-0.5cm]
	
	$^{b}$ Institute of Radiation Problems, Azerbaijan National Academy of Sciences, \\ B.Vahabzade St. 9, AZ1143, Baku, Azerbaijan\\[-0.5cm]
	
	$^{c}$ Department of Mathematics, Khazar University,  Mehseti St. 41, AZ1096, Baku, Azerbaijan \\[-0.5cm]
}

\emailAdd{ilmar.gahramanov@boun.edu.tr}
\emailAdd{erkan.kaluc@boun.edu.tr}

\abstract{In this work, we construct new Bailey pairs for the integral pentagon identity in terms of q-hypergeometric functions. The pentagon identity considered here represents equality of the partition functions of a certain three-dimensional supersymmetric dual theories. It can be also interpreted as the star-triangle relation for the Ising-type integrable lattice model.}


\begin{document}
	\maketitle
	\flushbottom

\section{Introduction}

Bailey’s lemma \cite{bailey1946some,bailey1948identities} is a powerful tool to derive hypergeometric identities (ordinary, trigonometric, and elliptic type). In this work, we construct new integral Bailey pairs for the pentagon identity in terms of q-hypergeometric functions. The pentagon identity can be interpreted as a Pachner's $3$--$2$ move for triangulated three-dimensional manifolds. Such identities also play a role in the study of supersymmetric gauge theories, integrable models, knot theory, etc.\footnote{See some recent works \cite{Kashaev:2012cz,Kashaev:2014rea,kashaev2014euler,Gahramanov:2013rda,Gahramanov:2014ona,Gahramanov:2016wxi,Bozkurt:2018xno,Bozkurt:2020gyy,Jafarzade:2018yei}.}

Let $q, z \in \mathbb{Z}$ with $|q| < 1$. We define the infinite q-product
\begin{align}
    (z;q)_{\infty}:=\prod^{\infty}_{k=0}(1-zq^k) \;.
\end{align}
We also adopt the following convention
\begin{align}
    (a,b;q)_{\infty}:=(a,q)_{\infty}(b,q)_{\infty} \; .
\end{align}
\begin{theorem}
Let $a_1, a_2, a_3, b_1, b_2, b_3, q \in \mathbb{C}$ and integers $m_i, n_i \in \mathbb{Z}$. Then
\begin{align} \nonumber
     &\sum_{m \in \mathbb{Z}} \int_{\mathbb{T}} \frac{dz}{2\pi iz} (-q^{\frac{1}{2}})^{\sum^{3}_{i=1}\frac{|m_i+m|}{2}+\frac{|n_i-m|}{2}}z^{-\sum^{3}_{i=1}(\frac{|m_i+m|}{2}-\frac{|n_i-m|}{2})} \\ \nonumber
     & \qquad \qquad \times \prod^{3}_{i=1}a_i^{-\frac{|m_i+m|}{2}}b_i^{-\frac{|n_i-m|}{2}}\frac{(q^{1+\frac{|m_i+m|}{2}}\frac{1}{a_iz},q^{1+\frac{|n_i-m|}{2}}\frac{z}{b_i};q)_{\infty}}{(q^{\frac{|m_i+m|}{2}}a_iz,q^{\frac{|n_i-m|}{2}}\frac{b_i}{z};q)_{\infty}}\\
     &=(-q^{\frac{1}{2}})^{\sum^{3}_{i,j=1}\frac{|m_i+n_j|}{2}}\prod^{3}_{i,j=1}(a_ib_j)^{-\frac{|m_i+n_j|}{2}}\frac{(q^{1+\frac{|m_i+n_j|}{2}}\frac{1}{a_ib_j};q)_{\infty}}{(q^{\frac{|m_i+n_j|}{2}}a_ib_j;q)_{\infty}}
     \label{relation}
     \;\text{,}
\end{align}

where the balancing conditions are
\begin{align}
\label{balance1}
\prod^{3}_{i=1}a_i b_i=q  \; \text{,}
\end{align}
\begin{align}
\sum^{3}_{i=1}m_i+n_i=0 \; \text{,}
\label{balance2}
\end{align} 
and $\mathbb{T}$ represents the positively oriented unit circle.
\end{theorem}

We would like to mention that the integral identity represents the supersymmetric duality for three-dimensional $\mathcal N=2$ supersymmetric gauge theories with the flavor symmetry\footnote{In this case parameters $a_i$ and $b_i$ stand for the flavor symmetry and $z$ is the fugacity for the $U(1)$ gauge group.}  $ SU(3) \times SU(3) \times U(1)$. This identity can also be written as the star-triangle relation\footnote{In this case parameters $a_i$, $b_i$ and $z$ stand for the continuous spin variables.} for some integrable model of statistical mechanics.

The proof of the form above is given in \cite{Gahramanov:2016wxi} for the balancing conditions\footnote{Yet, as $SU(3) \times SU(3) \times U(1)$ has five independent parameters, the above form must be correct even for the more general balancing conditions in (\ref{balance1},\ref{balance2}).}
\begin{align}
    \prod^{3}_{i=0}a_i= \prod^{3}_{i=0}b_i=q^{\frac{1}{2}} \; \text{,} \; \;\; \; \;\; \; \;  \sum^{3}_{i=0}m_i=\sum^{3}_{i=0}n_i=0 \; \text{.}
\end{align}

The absolute values can be eliminated by the identity \cite{Dimofte:2011py}
\begin{align}
    \frac{(q^{1+\frac{|m|}{2}}/z;q)_{\infty}}{(q^{\frac{|m|}{2}}z;q)_{\infty}}=(-q^{-\frac{1}{2}}z)^{\frac{|m|-m}{2}}\frac{(q^{1+\frac{m}{2}}/z;q)_{\infty}}{(q^{\frac{m}{2}}z;q)_{\infty}} \; \text{,}
\end{align}
and one ends up with the following $q$-hypergeometric sum/integral identity \cite{Gahramanov:2013rda,Gahramanov:2014ona,Gahramanov:2016wxi}
\begin{align}
    \sum_{m \in \mathbb{Z}} \int_{\mathbb{T}} \prod^{3}_{i=1}\frac{(q^{1+\frac{m+m_i}{2}}\frac{1}{a_iz},q^{1+\frac{n_i-m}{2}}\frac{z}{b_i};q)_{\infty}}{(q^{\frac{m+m_i}{2}}a_iz,q^{\frac{n_i-m}{2}}\frac{b_i}{z};q)_{\infty}}\frac{1}{z^{3m}}\frac{dz}{2\pi iz}=\frac{1}{\prod^{3}_{i=1}a^{m_i}_ib^{n_i}_i}\prod^{3}_{i,j=1}\frac{(q^{1+\frac{m_i+n_j}{2}}\frac{1}{a_ib_j};q)_{\infty}}{(q^{\frac{m_i+n_j}{2}}a_ib_j;q)_{\infty}} \;.
\end{align}

\section{Integral pentagon identity}

In \cite{Gahramanov:2016wxi,Gahramanov:2014ona,Gahramanov:2013rda} it was shown that the identity (\ref{relation}) can be written as an integral pentagon identity
\begin{align} \nonumber
    \sum_{m \in  \mathbb{Z}} \int_{\mathbb{T}} \frac{dz}{2\pi i z} &\prod^3_{i=1}\mathcal{B}[a_i, n_i+m; b_iz^{-1}, m_i-m]\\ \label{mainpent}
    =\mathcal{B}[a_1b_2, n_1+m_2; a_3b_1; n_3+m_1]&\mathcal{B}[a_2b_1, n_2+m_1; a_3b_2, n_3+m_2] \; \text{,}
\end{align}
where we define the following function as
\begin{align}
    \mathcal{B}_m[a,n;b,m]=(-q^{\frac{1}{2}})^{\frac{|n|}{2}+\frac{|m|}{2}+\frac{|n+m|}{2}}a^{-\frac{|n|}{2}}b^{-\frac{|m|}{2}}(ab)^{\frac{|n+m|}{2}} \times \frac{(q^{1+\frac{|n|}{2}}a^{-1}, q^{1+\frac{|m|}{2}}b^{-1}, q^{\frac{|n+m|}{2}}ab;q)_{\infty}}{(q^{\frac{|n|}{2}}a, q^{\frac{|m|}{2}}b, q^{1+\frac{|n+m|}{2}}(ab)^{-1};q)_\infty} \; \text{.}
\end{align}

In a general sense, any algebraic relation for operators $\mathcal B$
\begin{align}
    \mathcal{B}\mathcal{B}\mathcal{B}=\mathcal{B}\mathcal{B}
\end{align}
which can be interpreted as a 2-3 Pachner move of a triangulated three-dimensional manifold is called a  pentagon relation \cite{Kashaev:2014rea,kashaev2014euler}. Note that the integral pentagon identity (\ref{mainpent}) for the $\mathcal N=2$ supersymmetric $S^2 \times S^1$ partition functions is supposed to be related to some topological invariant of corresponding $3$–manifold via $3d$--$3d$ correspondence \cite{Dimofte:2011ju,Dimofte:2011py} that connects three-dimensional $\mathcal N = 2$ supersymmetric theories and triangulated 3–manifolds. There are several examples of pentagon identities arising from supersymmetric gauge theory computations, see, e.g. \cite{Dimofte:2011ju,Dimofte:2011py,Gang:2013sqa,Benvenuti:2016wet,Gahramanov:2013rda,Gahramanov:2014ona,Gahramanov:2016wxi,Bozkurt:2018xno,Jafarzade:2018yei,Bozkurt:2020gyy}.

\section{Bailey pairs}

Rogers-Ramanujan type identities are being continuously used in the solution of the integrable models, namely to derive the Yang-Baxter and the pentagon identities. In fact, a well-known example of this usage is conducted during the investigations of the hard hexagon model by Baxter. It turns out that Bailey discovered a systematic way to derive these types of identities \cite{bailey1946some, bailey1948identities,warnaar200150, 2019Zudilin}. As generalized by Andrews \cite{andrews1984multiple,andrews2001bailey}, there exists an iterative scheme to derive infinitely many of these identities if one pair, called a Bailey pair is known. This form the so-called Bailey chain. The induction step of generating the particular Bailey pairs is referred as the Bailey lemma for the chain we consider.

A generalization of the Bailey pairs approach to the integral identities is firstly done by Spiridonov in \cite{spiridonov2002elliptic,spiridonov2004bailey}. The construction of integral Bailey pairs yields new powerful verifications of various supersymmetric dualities \cite{Brunner:2016uvv,Brunner:2017lhb}, generating solutions to the Yang-Baxter equation \cite{Derkachov:2012iv,Gahramanov:2015cva,Magadov:2018hlc,Spiridonov:2019uuw}, etc.

Accordingly, the generalized version of the Bailey chain is a couple of infinite sequences of holomorphic functions $\{\alpha^{(i)}_n\}_{n \geq 0}$ and $\{\beta^{(i)}_n\}_{n \geq 0}$ such that there exists an identity independent of \textit{i} which connect $\alpha^{(i)}_n$ and $\beta^{(i)}_n$ as
\begin{equation}
    \beta^{(i)}_n=F_n(\alpha^{(i)}_0, \alpha^{(i)}_1,..., \alpha^{(i)}_n) \; \text{,}
\end{equation}
where $F$ can be an operator which may now include sum or integrals. Here, $\alpha^{(i)}_n$ and $\beta^{(i)}_n$ are constructed according to
\begin{align}
    \alpha^{(i)}_n=G(\alpha^{(i)}_0, \alpha^{(i)}_1,..., \alpha^{(i)}_{n-1}) \; \text{,}\\
    \beta^{(i)}_n=H(\beta^{(i)}_0, \beta^{(i)}_1,..., \beta^{(i)}_{n-1}) \; \text{,}
\end{align}
where $G$ and $H$ represent integral-sum operators.

\begin{definition}
Let $\{\alpha_m(z;t)\}_{m \in \mathbb{Z}}$ and $\{\beta_m(z;t)\}_{m \in \mathbb{Z}}$ be two sequences of functions. They are said to form a Bailey pair with respect to the parameter $t$ iff
\begin{align}
\label{bailey1}
    \beta_m(w;t)=\sum_{n \in \mathbb{Z}}\int dz \; \mathcal{B}[twz^{-1},m-n+n_t, tw^{-1}z,-m+n+n_t] \; \alpha_n(z;t) \; \text{.}
\end{align}

\end{definition}

\begin{lemma}
If $\{\alpha_m(z;t)\}_{m \in \mathbb{Z}}$ and $\{\beta_m(z;t)\}_{m \in \mathbb{Z}}$ form a Bailey pair with respect to $t$, then the following sequences
\begin{align}
     &\alpha'_n(w;st)=\mathcal{B}[tuw,n+n_u+n_t,s^2,2n_s]\alpha_n(w;t)\\ \nonumber
     &\beta'_n(w;st)=\sum_{m \in \mathbb{Z}}\int \frac{dx}{2\pi i x} \mathcal{B}[swx^{-1},-m+n+n_s;ux,n_u+m]\\ 
     & \qquad \qquad \qquad \qquad \times\mathcal{B}[st^{2}uw,n+2n_t+n_u+n_s,sw^{-1}x,-n+m+n_s]\beta_m(x;t)
     \label{bailey2}
\end{align}
\vspace{-0.5cm}
form a Bailey pair with respect to the parameter $st$.
\end{lemma}
\begin{proof}
We have to show that
\begin{align}
    \beta_n'(w,st)=\sum_{p \in \mathbb{Z}}\int \mathcal{B}[stwy^{-1},n-p+n_s+n_t,sty^{-1}x,-n+p+n_s+n_t]\alpha'_p(y,st)dy \;.
    \label{bailey3}
\end{align}

Inserting (\ref{bailey1}) in (\ref{bailey2}), we first calculate the left hand side of the equality (\ref{bailey3})
\begin{align} \nonumber
    &\beta'_n(w;st)=\sum_{m \in \mathbb{Z}}\int_{\mathbb{T}}\frac{dx}{2\pi i x} \mathcal{B}[swx^{-1},n+n_s-m,ux,m+n_u]\; \mathcal{B}[st^2uw,n+n_u+2n_t+n_s,sw^{-1}x,m-n+n_s]\\ \nonumber
    & \qquad \qquad \qquad \qquad \times \sum_{p \in \mathbb{Z}} \int_{\mathbb{T}}\mathcal{B}[txy^{-1},m-p+n_t,tx^{-1}y,-m+p+n_t]\alpha_p(y,t)dy\\ \nonumber
    &=\sum_{m \in \mathbb{Z}}\sum_{p \in \mathbb{Z}}\int \int \mathcal{B}[swx^{-1},-m+n+n_s,ux,m+n_u] \; \mathcal{B}[st^2uw,n+n_u+2n_t+n_s,sw^{-1}x,-n+m+n_s]\\
     & \qquad \qquad \qquad \qquad \times\mathcal{B}[txy^{-1},m-p+n_t,tx^{-1}y,-m+p+n_t]\alpha_p(y,t)dy\frac{dx}{2\pi i x}
\end{align}

Hence, by regrouping the terms accordingly, we obtain\footnote{For convenience $q$ of the $q$-product is omitted.}
\begin{align} \nonumber
    &\sum_{p \in \mathbb{Z}} \sum_{m \in \mathbb{Z}} \int (-q^{\frac{1}{2}})^{\frac{|m+n_u|}{2}+\frac{|m-n+n_s|}{2}+\frac{|m-p+n_t|}{2}+\frac{|n-m+n_s|}{2}+\frac{|m-n_u|}{2}+\frac{|p-m+n_t|}{2}}\\ \nonumber
    & \times (ux)^{-\frac{|m+n_u|}{2}}(sw^{-1}x)^{-\frac{|m-n+n_s|}{2}}(ty^{-1}x)^{-\frac{|m-p+n_t|}{2}} 
    (swx^{-1})^{-\frac{|n-m+n_s|}{2}}(s^{2}t^{2}q^{-1}ux)^{\frac{|m-n_u|}{2}}(tyx^{-1})^{-\frac{|p-m+n_t|}{2}}\\ \nonumber
    & \times     \frac{(q^{1+\frac{|n-m+n_s|}{2}}(swx^{-1})^{-1})_{\infty}}{(q^{\frac{|n-m+n_s|}{2}}swx^{-1})_{\infty}}\frac{( q^{1+\frac{|m+n_u|}{2}}(ux)^{-1} )_{\infty}}{(q^{\frac{|m+n_u|}{2}}ux, )_\infty}\frac{(q^{1+\frac{|m-n+n_s|}{2}}(sw^{-1}x)^{-1})_{\infty}}{(q^{\frac{|m-n+n_s|}{2}}sw^{-1}x)_{\infty}}\\ \nonumber
    &\times \frac{(q^{1+\frac{|m-n_u|}{2}}s^{2}t^{2}q^{-1}ux)_{\infty}}{(q^{\frac{|m-n_u|}{2}}s^{-2}t^{-2}qu^{-1}x^{-1})_\infty}\frac{(q^{1+\frac{|m-p+n_t|}{2}}(ty^{-1}x)^{-1})_{\infty}}{(q^{\frac{|m-p+n_t|}{2}}ty^{-1}x)_{\infty}}\frac{(q^{1+\frac{|p-m+n_t|}{2}}(tyx^{-1})^{-1})_{\infty}}{( q^{\frac{|p-m+n_t|}{2}}tyx^{-1})_{\infty}}\\ \nonumber
    &\times (-q^{\frac{1}{2}})^{-\frac{|n-n_t|}{2}+|n_t|+\frac{|n+n_t|}{2}}(swu)^{\frac{|n-n_t|}{2}}(st^2uw)^{-\frac{|n+n_t|}{2}} (q^{-1}t^2)^{|n_t|}\\
    &\times \frac{(q^{1+\frac{|n-n_t|}{2}}swq^{-1}u)_{\infty}}{(q^{\frac{|n-n_t|}{2}}(swq^{-1}u)^{-1})_{\infty}}\frac{ (q^{1+\frac{|n+n_t|}{2}}(st^2uw)^{-1})}{(q^{\frac{|n+n_t|}{2}}st^2uw)_{\infty}}\frac{(q^{1+|n_t|}q^{-1}t^2)_{\infty}}{(q^{|n_t|}qt^{-2})_\infty}\alpha_p(y,t)dy\frac{dx}{2\pi i x}
\end{align}
where we required the sum of the powers of $x$ to vanish, namely
\begin{align}
    n_u+n_s+n_t=0 \label{constraint}
\end{align}

Upon renaming the variables as
  \begin{align}
        &a_1=u \to m_1=n_u \qquad \qquad \qquad \qquad b_1=sw \to n_1=n+n_s\\
        &a_2=sw^{-1} \to m_2=-n+n_s \; \; \;\; \;\; \;\; \;\; \; \;\; \; b_2=qs^{-2}t^{-2}u^{-1} \to n_2=n_u\\
        &a_3=ty^{-1} \to m_3=-p+n_t \qquad \;\;\; \;\; \; \;\; \; b_3=tx \to n_3=p+n_t
    \end{align}
we identify the integral relation (\ref{relation}). Also observe that the constraint (\ref{constraint}) resulted in the balancing condition (\ref{balance2}). We hence get upon simplification and regrouping of the terms
\begin{align}
    \nonumber
    &\sum_{p \in \mathbb{Z}}\int \alpha_p(y,t)dy(-q^{\frac{1}{2}})^{\frac{|n-p-n_u|}{2}+\frac{|-n+p-n_u|}{2}-|n_u|}(stwy^{-1})^{-\frac{|n-p-n_u|}{2}}(stw^{-1}y)^{-\frac{|p-n-n_u|}{2}}(s^2t^2)^{|n_u|}\\
     \nonumber
    &\times \frac{(q^{1+\frac{|n-p-n_u|}{2}}(stwy^{-1})^{-1})_{\infty}}{(q^{\frac{|n-p-n_u|}{2}}stwy^{-1})_{\infty}}\frac{(q^{1+\frac{|p-n-n_u|}{2}}(stw^{-1}y)^{-1})_{\infty}}{(q^{|p-n-n_u|}stw^{-1}y)_{\infty}}\frac{(q^{\frac{|n_u|}{2}}s^2t^2)_{\infty}}{(q^{1+\frac{|n_u|}{2}}s^{-2}t^{-2})_{\infty}}\\
     \nonumber
    &\times (-q^{\frac{1}{2}})^{\frac{|p-n_s|}{2}+|n_s|-\frac{|n_s+p|}{2}}(tuy)^{-\frac{|p-n_s|}{2}}(s^2)^{-|n_s|}(s^2tuy)^{\frac{|p+n_s|}{2}}\\
        &\times \frac{(q^{1+\frac{|p-n_s|}{2}}(tuy)^{-1})_{\infty}}{(q^{\frac{|p-n_s|}{2}}tuy)_{\infty}}\frac{(q^{1+|n_s|}s^{-2})_{\infty}}{(q^{|n_s|}s^2)_{\infty}}\frac{(q^{1+\frac{|p+n_s|}{2}}s^2tuy)_{\infty}}{(q^{\frac{|p+n_s|}{2}}(s^{2}tuy)^{-1})_{\infty}}\; \text{,}
\end{align}
which is the desired operator equality
\begin{align}
 \nonumber
     &\sum_{p \in \mathbb{Z}} \int dy\mathcal{B}[stwy^{-1},n_s+n_t+n-p,stw^{-1}y,n_s+n_t-n+p)\mathcal{B}[tyu,n_t+p+n_u,s^2,2n_s]\\
    &=\sum_{p \in \mathbb{Z}}\int dy \mathcal{B}[stwy^{-1},n_s+n_t+n-p,stw^{-1}y,n_s+n_t-n+p)
    \; \text{.}
\end{align}

\end{proof}

\section{Conclusions}

In this work, we constructed new integral Bailey pairs for pentagon identity in terms of $q$-hypergeometric functions. One can use this construction to obtain new supersymmetric dualities for quiver theories. 

We would like to mention that the pentagon identity presented here can also be written as the star-triangle relation for some integrable lattice model of statistical mechanics. It would be interesting to construct the Bailey pairs corresponding to the star-triangle form of the same integral identity. 

\section*{Acknowledgements} 

The authors are grateful to Mustafa Mullahasanoglu, Sena Ergisi, Dilara Kosva, and Batuhan Keskin for valuable discussions related to the construction of the Bailey pairs and pentagon identities. The work of Ilmar Gahramanov is partially supported by the Bogazici University Research Fund under grant number 20B03SUP3 and TUBITAK grant 220N106.

\bibliographystyle{utphys}
\bibliography{pentagon}

\providecommand{\href}[2]{#2}\begingroup\raggedright\begin{thebibliography}{10}

\bibitem{bailey1946some}
W.~N. Bailey, ``Some identities in combinatory analysis,''
  \href{http://dx.doi.org/10.1112/plms/s2-49.6.421}{{\em Proceedings of the
  London Mathematical Society} {\bfseries 2} no.~1, (1946) 421--435}.

\bibitem{bailey1948identities}
W.~N. Bailey, ``{Identities of the Rogers-Ramanujan Type},''
  \href{http://dx.doi.org/10.1112/plms/s2-50.1.1}{{\em Proceedings of the
  London Mathematical Society} {\bfseries 2} no.~1, (1948) 1--10}.

\bibitem{Kashaev:2012cz}
R.~Kashaev, F.~Luo, and G.~Vartanov, ``{A TQFT of Turaev\textendash{}Viro Type
  on Shaped Triangulations},''
  \href{http://dx.doi.org/10.1007/s00023-015-0427-8}{{\em Annales Henri
  Poincare} {\bfseries 17} no.~5, (2016) 1109--1143},
  \href{http://arxiv.org/abs/1210.8393}{{\ttfamily arXiv:1210.8393 [math.QA]}}.

\bibitem{Kashaev:2014rea}
R.~M. Kashaev, ``{Beta pentagon relations},''
  \href{http://dx.doi.org/10.1007/s11232-014-0208-4}{{\em Theor. Math. Phys.}
  {\bfseries 181} no.~1, (2014) 1194--1205},
  \href{http://arxiv.org/abs/1403.1298}{{\ttfamily arXiv:1403.1298 [math-ph]}}.

\bibitem{kashaev2014euler}
R.~Kashaev, ``Euler’s beta function and pentagon relations,''
  \href{http://dx.doi.org/10.1007/s40306-014-0080-1}{{\em Acta Mathematica
  Vietnamica} {\bfseries 39} no.~4, (2014) 561--566}.

\bibitem{Gahramanov:2013rda}
I.~Gahramanov and H.~Rosengren, ``{A new pentagon identity for the tetrahedron
  index},'' \href{http://dx.doi.org/10.1007/JHEP11(2013)128}{{\em JHEP}
  {\bfseries 11} (2013) 128},
\href{http://arxiv.org/abs/1309.2195}{{\ttfamily arXiv:1309.2195 [hep-th]}}.

\bibitem{Gahramanov:2014ona}
I.~Gahramanov and H.~Rosengren, ``{Integral pentagon relations for 3d
  superconformal indices},'' \href{http://arxiv.org/abs/1412.2926}{{\ttfamily
  arXiv:1412.2926 [hep-th]}}.
[Proc. Symp. Pure Math.93,165(2016)].

\bibitem{Gahramanov:2016wxi}
I.~Gahramanov and H.~Rosengren, ``{Basic hypergeometry of supersymmetric
  dualities},'' \href{http://dx.doi.org/10.1016/j.nuclphysb.2016.10.004}{{\em
  Nucl. Phys.} {\bfseries B913} (2016) 747--768},
\href{http://arxiv.org/abs/1606.08185}{{\ttfamily arXiv:1606.08185 [hep-th]}}.

\bibitem{Bozkurt:2018xno}
D.~N. Bozkurt and I.~Gahramanov, ``{Pentagon identities arising in
  supersymmetric gauge theory computations},''
  \href{http://dx.doi.org/10.1134/S0040577919020028}{{\em Teor. Mat. Fiz.}
  {\bfseries 198} no.~2, (2019) 215--224},
\href{http://arxiv.org/abs/1803.00855}{{\ttfamily arXiv:1803.00855 [math-ph]}}.

\bibitem{Bozkurt:2020gyy}
D.~N. Bozkurt, I.~Gahramanov, and M.~Mullahasanoglu, ``{Lens partition
  function, pentagon identity, and star-triangle relation},''
  \href{http://dx.doi.org/10.1103/PhysRevD.103.126013}{{\em Phys. Rev. D}
  {\bfseries 103} no.~12, (2021) 126013},
  \href{http://arxiv.org/abs/2009.14198}{{\ttfamily arXiv:2009.14198
  [hep-th]}}.

\bibitem{Jafarzade:2018yei}
S.~Jafarzade, ``{New Pentagon Identities Revisited},''
  \href{http://dx.doi.org/10.1088/1742-6596/1194/1/012054}{{\em J. Phys. Conf.
  Ser.} {\bfseries 1194} no.~1, (2019) 012054},
\href{http://arxiv.org/abs/1812.01325}{{\ttfamily arXiv:1812.01325 [math-ph]}}.

\bibitem{Dimofte:2011py}
T.~Dimofte, D.~Gaiotto, and S.~Gukov, ``{3-Manifolds and 3d Indices},''
  \href{http://dx.doi.org/10.4310/ATMP.2013.v17.n5.a3}{{\em Adv. Theor. Math.
  Phys.} {\bfseries 17} no.~5, (2013) 975--1076},
  \href{http://arxiv.org/abs/1112.5179}{{\ttfamily arXiv:1112.5179 [hep-th]}}.

\bibitem{Dimofte:2011ju}
T.~Dimofte, D.~Gaiotto, and S.~Gukov, ``{Gauge Theories Labelled by
  Three-Manifolds},'' \href{http://dx.doi.org/10.1007/s00220-013-1863-2}{{\em
  Commun. Math. Phys.} {\bfseries 325} (2014) 367--419},
  \href{http://arxiv.org/abs/1108.4389}{{\ttfamily arXiv:1108.4389 [hep-th]}}.

\bibitem{Gang:2013sqa}
D.~Gang, E.~Koh, S.~Lee, and J.~Park, ``{Superconformal Index and 3d-3d
  Correspondence for Mapping Cylinder/Torus},''
  \href{http://dx.doi.org/10.1007/JHEP01(2014)063}{{\em JHEP} {\bfseries 01}
  (2014) 063}, \href{http://arxiv.org/abs/1305.0937}{{\ttfamily arXiv:1305.0937
  [hep-th]}}.

\bibitem{Benvenuti:2016wet}
S.~Benvenuti and S.~Pasquetti, ``{3d $ \mathcal{N} $ = 2 mirror symmetry,
  pq-webs and monopole superpotentials},''
  \href{http://dx.doi.org/10.1007/JHEP08(2016)136}{{\em JHEP} {\bfseries 08}
  (2016) 136}, \href{http://arxiv.org/abs/1605.02675}{{\ttfamily
  arXiv:1605.02675 [hep-th]}}.

\bibitem{warnaar200150}
S.~O. Warnaar, ``50 years of bailey’s lemma,''
  \href{http://dx.doi.org/10.1007/978-3-642-59448-9_23}{{\em Algebraic
  combinatorics and applications} (2001) 333--347},
  \href{http://arxiv.org/abs/0910.2062}{{\ttfamily arXiv:0910.2062 [math.CO]}}.

\bibitem{2019Zudilin}
W.~Zudilin, ``{Hypergeometric heritage of W. N. Bailey},''
  \href{http://dx.doi.org/10.4310/iccm.2019.v7.n2.a4}{{\em Notices of the
  International Congress of Chinese Mathematicians} {\bfseries 7} no.~2, (2019)
  32–46}.

\bibitem{andrews1984multiple}
G.~Andrews, ``Multiple series rogers-ramanujan type identities,''
  \href{http://dx.doi.org/10.2140/pjm.1984.114.267}{{\em Pacific journal of
  mathematics} {\bfseries 114} no.~2, (1984) 267--283}.

\bibitem{andrews2001bailey}
G.~E. Andrews, ``Bailey’s transform, lemma, chains and tree,''
  \href{http://dx.doi.org/10.1007/978-94-010-0818-1_1}{{\em Special Functions
  2000: Current Perspective and Future Directions} (2001) 1--22}.

\bibitem{spiridonov2002elliptic}
V.~Spiridonov, ``An elliptic incarnation of the bailey chain,''
  \href{http://dx.doi.org/10.1155/S1073792802205127}{{\em International
  Mathematics Research Notices} {\bfseries 2002} no.~37, (2002) 1945--1977}.

\bibitem{spiridonov2004bailey}
V.~P. Spiridonov, ``A bailey tree for integrals,''
  \href{http://dx.doi.org/10.1023/B:TAMP.0000022745.45082.18}{{\em Theoretical
  and mathematical physics} {\bfseries 139} no.~1, (2004) 536--541},
  \href{http://arxiv.org/abs/math/0312502}{{\ttfamily arXiv:math/0312502}}.

\bibitem{Brunner:2016uvv}
F.~Br\"unner and V.~P. Spiridonov, ``{A duality web of linear quivers},''
  \href{http://dx.doi.org/10.1016/j.physletb.2016.08.039}{{\em Phys. Lett. B}
  {\bfseries 761} (2016) 261--264},
  \href{http://arxiv.org/abs/1605.06991}{{\ttfamily arXiv:1605.06991
  [hep-th]}}.

\bibitem{Brunner:2017lhb}
F.~Br\"unner and V.~P. Spiridonov, ``{4d $\mathcal{N}=1$ quiver gauge theories
  and the $\mathrm{A_n}$ Bailey lemma},''
  \href{http://dx.doi.org/10.1007/JHEP03(2018)105}{{\em JHEP} {\bfseries 03}
  (2018) 105}, \href{http://arxiv.org/abs/1712.07018}{{\ttfamily
  arXiv:1712.07018 [hep-th]}}.

\bibitem{Derkachov:2012iv}
S.~E. Derkachov and V.~P. Spiridonov, ``{Yang-Baxter equation, parameter
  permutations, and the elliptic beta integral},''
  \href{http://dx.doi.org/10.1070/RM2013v068n06ABEH004869}{{\em Russ. Math.
  Surveys} {\bfseries 68} (2013) 1027--1072},
\href{http://arxiv.org/abs/1205.3520}{{\ttfamily arXiv:1205.3520 [math-ph]}}.

\bibitem{Gahramanov:2015cva}
I.~Gahramanov and V.~P. Spiridonov, ``{The star-triangle relation and 3d
  superconformal indices},''
  \href{http://dx.doi.org/10.1007/JHEP08(2015)040}{{\em JHEP} {\bfseries 08}
  (2015) 040},
\href{http://arxiv.org/abs/1505.00765}{{\ttfamily arXiv:1505.00765 [hep-th]}}.

\bibitem{Magadov:2018hlc}
K.~Y. Magadov and V.~P. Spiridonov, ``{Matrix Bailey Lemma and the
  Star-Triangle Relation},''
  \href{http://dx.doi.org/10.3842/SIGMA.2018.121}{{\em SIGMA} {\bfseries 14}
  (2018) 121}.

\bibitem{Spiridonov:2019uuw}
V.~P. Spiridonov, ``{The rarefied elliptic Bailey lemma and the
  Yang\textendash{}Baxter equation},''
  \href{http://dx.doi.org/10.1088/1751-8121/ab3358}{{\em J. Phys. A} {\bfseries
  52} no.~35, (2019) 355201}, \href{http://arxiv.org/abs/1904.12046}{{\ttfamily
  arXiv:1904.12046 [math-ph]}}.

\end{thebibliography}\endgroup

\end{document}